 \documentclass[letterpaper, 10 pt, journal,twoside]{IEEEtran}
\IEEEoverridecommandlockouts
\usepackage{latexsym,amsfonts,amssymb,amsmath,amscd}
\usepackage{graphicx}    % Necessary to use \scalebox
\usepackage{mathtools, nccmath}
  \usepackage{amsthm}
\newtheorem{theorem}{Theorem}[section]

\newtheorem{lemma}[theorem]{Lemma}

\newtheorem{proposition}[theorem]{Proposition}
\newtheorem{Corollary}[theorem]{Corollary}

\usepackage{subfigure}
\usepackage{atbegshi}
\allowdisplaybreaks
\usepackage{cite}
\usepackage{mathrsfs}
\usepackage{amsmath,amssymb,amsfonts}
\usepackage{xcolor}
\graphicspath{ {./Pictures/} }
\def\BibTeX{{\rm B\kern-.05em{\sc i\kern-.025em b}\kern-.08em
    T\kern-.1667em\lower.7ex\hbox{E}\kern-.125emX}}
\title{Exact solution and projection filters for open quantum systems subject to imperfect measurements}
\author{Ibrahim Ramadan, Nina H. Amini, and Paolo Mason$^{*}$% <-this % stops a space
\thanks{$^{*}$The authors are with Laboratoire des signaux et syst\`emes (L2S)
CNRS, CentraleSup\'elec, Universit\'e Paris-Sud, Universit\'e Paris-Saclay.
3, rue Joliot-Curie, 91192 Gif-sur-Yvette, cedex, France
        {\tt\small firstname.lastname@l2s.centralesupelec.fr}}
}

\begin{document}

\maketitle

\thispagestyle{empty}

\begin{abstract}
%In this paper, we consider an open quantum system undergoing imperfect and indirect measurement. The main goal is to develop an approximate quantum filter   by  projecting on a lower dimensional submanifold in view of reducing the computational cost of real-time simulations of the quantum filter. We generalize  the approach of \cite{Gao}, originally developed for perfect measurement, %developed  in \cite{Gao} for perfect measurement 
%in presence of detection imperfection. An error analysis is performed to evaluate the precision of this approximate quantum filter, focusing on the case of quantum non-demolition (QND) measurement. Furthermore, for QND measurement, we show that this projection approach allows to express the exact solution of the quantum filter equation in terms of  the solution of a lower dimensional stochastic differential equation. 
%Simulation suggests the efficiency of the proposed
%quantum projection filter, even in  presence of a stabilizing feedback control which depends on the projection filter.
In this paper, we consider an open quantum system undergoing imperfect and indirect measurement. For quantum non-demolition (QND) measurement, we show that the system evolves on an appropriately chosen manifold and we express the exact solution of the quantum filter equation in terms of  the solution of a lower dimensional stochastic differential equation. In order to further reduce the dimension of the system under study, we consider the projection on the lower dimensional manifold originally introduced in \cite{Gao} for the case of perfect measurements. An error analysis is performed to evaluate the precision of this approximate quantum filter, focusing on the case of QND measurement. Simulations suggest the efficiency of the proposed
quantum projection filter, even in  presence of a stabilizing feedback control which depends on the projection filter.
\end{abstract}
\section{Introduction}
%Similar to the crucial role of classical control theory in technologies, quantum control theory is central in development of quantum technologies. 
Being able to reliably control quantum dynamics is a fundamental step towards the development of quantum technologies.
Quantum systems may be assumed to be in closed or open form. Unlike closed systems, open quantum systems are by definition in interaction with an environment, hence they provide a more realistic description of physical systems. On the other hand, the interaction with the environment entails decoherence phenomena, %i.e., 
characterized by a loss of information \cite{davies1976quantum}. For controlled open quantum systems, closed-loop control strategies are preferable, compared to open-loop ones, due to robustness issues. A measurement-based feedback strategy can be realized based on an estimation of the state which is  obtained by partial observations of the system. Such an estimation is called quantum filter or quantum trajectory in physics literature \cite{belavkin1989nondemolition,bouten2007introduction,barchielli2009quantum,wiseman2009quantum}. The controlled dynamics obtained in this way %problem 
fits in the framework of stochastic control, see e.g., \cite{belavkin2004towards} for further clarifications. %One of the obstacles to realize a feedback strategy in real experiments is the computation time to simulate the quantum filter. 

The works in e.g.,~\cite{Vanhandel3,liang2019exponential} present feedback stabilization of some particular open quantum systems by using geometric control, Lyapunov methods, and stochastic tools. Feedback stabilization methods are based on the real-time simulation of a quantum filter equation to obtain an estimate of the quantum state. The evolution of the quantum filter is usually described by a large number of equations and their simulation represents an obstacle to realize in real-time feedback strategies in real experiments. For instance, for a $m$-qubit system, the evolution of the density matrix is described by $2^{2m}-1$ stochastic differential equations. 
%As this quantum filter equation is usually high-dimensional, this, in general, makes it impossible (due to the short time scales of the quantum system) to achieve real-time simulation of the filter equation. 
As in the classical case, in order to tackle this issue the basic idea is to seek reduced dynamics containing enough information in order to design efficient feedback controls based on them. Such feedback strategies should possibly be robust with respect to experimental imperfections. %filter equation 
 For instance, in \cite{liang2022model}, the authors show the robustness of a stabilizing feedback depending on %an approximation of
 a reduced dynamics only involving the diagonal elements of the filter state  in the case of QND measurements.

%The strategy here, as in the classical case, is to seek a reduced filter equation. The main point is that such reduced filters should have enough accuracy to be applied in the design of the feedbacks, and such reduced filters should not perturb the efficiency and robustness of the feedback strategies with respect to experimental imperfections. For instance, in \cite{liang2022model}, the authors show the robustness of a stabilizing feedback depending on an approximation of diagonal elements of the filter state  in the case of QND measurement. 

The projection filtering strategy has been developed in the classical case in~\cite{Amari,Brigo1,Brigo2},  based on differential and information geometry tools. To our knowledge, the quantum projection filter scheme was first proposed in \cite{Vanhandel1}. Later,  in \cite{Nielsen} the authors obtained the evolution of system state in a lower dimensional manifold by unsupervised learning. This was  achieved by use of local tangent space alignment. In \cite{tezak2017low}, a dynamical law is derived by minimizing the statistical distance in the moving basis and an equivalence with the projection filter has been shown. Recently, in  \cite{Gao}, a quantum projection filtering approach was developed in which the dynamics is projected onto a manifold consisting of an exponential family of unnormalized density matrices. An extended Kalman filter and numerical approaches have been respectively established in  \cite{Emzir} and  \cite{Rouchon}. %In the aim of feedback stabilization, the main goal is to construct feedbacks based on reduced filters stabilizing any open quantum systems. This will be more realistic to be implemented in real experiments. 

In this paper, we consider an open quantum system undergoing indirect measurement in presence of detection imperfections. Firstly by suitably choosing a submanifold of the state space, we show that the exact solution $\rho_t$ of the quantum filter equation under QND measurement can be expressed in parametrized form as $\rho_{\phi_t},$ where $\phi_t$ corresponds to  the solution of a lower dimensional stochastic differential equation.  Note that similar results have been derived for the particular case of qubit systems, with a different approach, in~\cite{rouchon-sarlette}. Then, in order to further reduce the complexity of the dynamics, i.e., to reduce the dimension of the parameter $\phi,$ we follow the projection filter approach introduced in \cite{Gao}, originally developed for perfect measurements. Specifically we adapt the computation of the approximation error in the case of  imperfect measurements. We observe that under QND measurements, the asymptotic behavior of the approximate projection filter is compatible with the original filter, in the sense that both dynamics converge to the set of invariant subspaces. This motivates the  application of a projection filter in a stabilizing feedback control law. To this aim, we verify numerically the efficiency of the stabilizing feedback control introduced in \cite{Liang} in the case of a two-level quantum system evaluated at the approximate filter.  This is promising for further investigations.

%This is found by projecting the dynamics on a suitable manifold.  

%Then, by introducing a new exponential submanifold, we obtain the exact solution of the quantum filter under QND measurement. Notably, the latter is obtained in terms of the solution of  a lower dimensional stochastic differential equation, compared to the original filter. Furthermore, we adapt the computation of the approximation error in the case of imperfections by projecting on the exponential submanifold introduced in \cite{Gao}. Moreover, by the help of the stochastic Lasalle-theorem \cite{Mao}, we prove analytically the convergence of the projected filter to its equilibrium state. This proofs the strength and importance of our projection process, not just for dimension reduction, but also for the behavior analysis. 
%A quantum projection filtering approach may be of particular interest in the stabilizing feedback design, provided that the designed feedback law still stabilizes the system whenever the quantum filter is replaced by its approximation.
%The main perspective of this paper is the use of such approximated filters in the stabilizing feedback design. 
%Here, as a first step to study this problem, we verify numerically the efficiency of a stabilizing feedback control when evaluated at the proposed approximated filter state. Numerical simulation in the case of a two-level quantum system by using the feedback control proposed in \cite{Liang} appears promising for further investigations.  

This paper is organized as follows. Section \ref{P} introduces the quantum filter equation under consideration.  Section \ref{exact} is devoted to the study of its exact solution in the case of QND measurements. In Section \ref{Q}, we develop the projection filter approach in the case of detection imperfections and we characterize the residual errors obtained from the projection process. We also derive a bound on the average total residual norm by assuming QND measurement. Also, we obtain a quantum state reduction result for the projection filter in the case of QND measurements. In Section \ref{R}, we perform a numerical simulation  for the case of a two-level system and discuss the application of the projection filter in the feedback design suggested in \cite{Liang}. 
Section \ref{S} provides a summary and gives some future perspectives.

{\bf{Notation.}} %Let $i=\sqrt{-1}$ be the imaginary unit.
%The complex conjugate transpose of a matrix $A$ is represented by $A^{\dag}$.%
% The identity matrix is displayed by $\mathbb{I}$. The trace of the square matrix $A$ is denoted by ${\rm Tr}(A)$.%
The singular values of $m\times m$ matrix $A$ are denoted by $s_{1}(A) \geq s_{2}(A)\geq\dots\geq s_{m}(A)$.
The commutator of matrices $A$ and $B$ is denoted by $[A,B] = AB-BA$. %
A square matrix $A$ is said to be Hermitian if $A=A^\dag,$ where $A^{\dag}$ corresponds to the complex conjugate transpose of $A.$ The Frobenius norm of $A$ is defined by $\lVert  A  \rVert_{F}=\sqrt{{\rm Tr}(A^{\dag}A)}.$
%\section{Preliminaries}\label{P}
\section{System description}\label{P}
Let us consider a finite dimensional open quantum system undergoing indirect measurement in the case of homodyne detection. 
%Let us denote by $\mathcal{H}$ the Hilbert space of the considered open quantum system. Suppose that dim($\mathcal{H}$) $= n < \infty$. %An atomic ensemble coupled to a cavity is 
%continuously interacting with an external single-channel laser 
%field in the vacuum state. 
The evolution of such a system is described by the following matrix-valued stochastic differential equation
\begin{align}
&d {\rho}_{t}=-i\left[H, {\rho}_{t}\right] dt
+ \left(L {\rho}_{t} L^{\dagger}-\frac{1}{2}(L^{\dagger} L \rho_{t}+{\rho}_{t} L^{\dagger} L)\right) dt \nonumber\\ 
&+\sqrt{\eta}\left(L {\rho}_{t}+{\rho}_{t} L^{\dagger}-\operatorname{Tr}\left[\left(L+L^{\dagger}\right){\rho}_{t}\right]{\rho}_{t}\right) dW_{t}.
 \label{1}
\end{align}
The density operator $\rho$ belongs to the space $S$ of Hermitian, positive semidefinite operators of trace one acting on  $\mathbb C^n.$

In the above equation, $H=H^{\dagger}$ is the Hamiltonian, $L$ represents the coupling operator, $0<\eta \leq 1$ is the detector efficiency. The classical Wiener process $W_{t}$    is related to the observation process ${Y}_{t},$ which is a continuous semimartingale with quadratic variation $\langle Y, Y\rangle_{t}=t$ satisfying
 \begin{align}
d{Y}_{t}=d{W}_{t}+\sqrt{\eta} \operatorname{Tr}\left[\left(L+L^{\dagger}\right) {\rho}_{t}\right] dt. \label{eq:y}\end{align}
Note that, for more general observation processes, the diffusion term in the evolution equation may be driven by a complex Wiener process (for more details see, e.g., \cite{Doherty}).

In the following we will mainly work with the Zakai equation, which is the  unnormalized form of the quantum filter equation \eqref{1} and which is given by %on it compared with the nonlinear filter equation
%in \eqref{1}. Consider
\begin{align}
d \tilde{\rho}_{t}&=-i\left[H, \tilde{\rho}_{t}\right]d t+\left(L \tilde{\rho}_{t} L^{\dagger}-\frac{1}{2}\left(L^{\dagger} L \tilde{\rho}_{t}+\tilde{\rho}_{t} L^{\dagger} L\right)\right) d t\nonumber\\
&+\sqrt{\eta}\left(L \tilde{\rho}_{t}+\tilde{\rho}_{t} L^{\dagger}\right) d Y_{t}. \label{zakai}
\end{align}
In particular %$\tilde{\rho}_{t}$ is such that 
$\rho_{t} = \tilde{\rho}_{t}/\mathrm{Tr}(\tilde{\rho}_{t})$. %\footnote{\color{red} We already said that it is unnormalized. Hence I would write: In particular $\rho_{t} = \tilde{\rho}_{t}/\mathrm{Tr}(\tilde{\rho}_{t})$} %The state $\tilde{\rho}_{t}$ is initially set to be $\tilde{\rho}_{0}= \rho_{0}$. 
Letting $\mathcal{A}$ be the set of all Hermitian operators on $\mathbb C^n,$  the evolution corresponding to~\eqref{zakai} takes place on the space
\begin{equation}
\mathcal{Q} = \{\rho \in \mathcal{A} \mid \rho \geq 0\},
\end{equation}
which is the closed subset of $\mathcal{A}$ consisting of all nonnegative Hermitian operators on $\mathbb C^n.$  %that contains the nonnegative Hermitian operators.
In particular $\mathcal{Q}$ can be seen as a differential manifold of dimension $n^2$. We denote by $T_{\rho}\mathcal{Q}$, the tangent space of $\mathcal{Q}$ at the point $\rho$, which is identified with $\mathcal{A}$.

 Since the vector fields defining the dynamics are linear, hence globally Lipschitz, Equation~\eqref{zakai} has a unique solution~\cite{Protter}. Since $\rho_{t} = \tilde{\rho}_{t}/\mathrm{Tr}(\tilde{\rho}_{t})$, we deduce that~\eqref{1} has a unique solution as well. %\footnote{\color{red} strange to say this for \eqref{zakai} and not for \eqref{1}}

For compatibility reasons with the differential manifold structure (see e.g.~\cite{Brigo1}), we further consider the Stratonovich form of the above equation, which is given by 
% The following Stratonovich quantum stochastic differential equation is equivalent to the It$\hat{o}$ quantum stochastic differential equation in \eqref{2}:
\begin{equation}
d \tilde{\rho}_{t}\!=\!\left(-i\left[H, \tilde{\rho}_{t}\right]+F \left(\tilde{\rho}_{t}\right)\right) d t+\sqrt{\eta}\left(L \tilde{\rho}_{t}+\tilde{\rho}_{t} L^{\dagger}\right) \circ d Y_{t}, \label{stra}
\end{equation} 

%\begin{equation}
%d \tilde{\rho}_{t}\!=\!\left(-i\left[H, \tilde{\rho}_{t}\right]-F %\left(\tilde{\rho}_{t}\right)\right) d t+\sqrt{\eta}\left(L %\tilde{\rho}_{t}+\tilde{\rho}_{t} L^{\dagger}\right) \circ d Y_{t}, \label{stra}
%\end{equation} 

where 
$F\left(\tilde{\rho}_{t}\right)\!=\!(1-\eta)L \tilde{\rho}_{t} L^{\dagger}-\frac{\left(\eta L+L^{\dagger}\right) L \tilde{\rho}_{t}+\tilde{\rho}_{t} L^{\dagger}\left(L+\eta L^{\dagger}\right)}{2}.$

%$F\left(\tilde{\rho}_{t}\right)\!=\!(\eta-1)L \tilde{\rho}_{t} L^{\dagger}+\frac{\left(\eta L+L^{\dagger}\right) L \tilde{\rho}_{t}+\tilde{\rho}_{t} L^{\dagger}\left(L+\eta L^{\dagger}\right)}{2}.$
%}

\section{Exact Solution}
\label{exact}
In this section, under suitable assumptions, we construct a submanifold $\mathcal M$ of $\mathcal A$ such that the dynamics given by \eqref{stra}, with initial condition $\rho_0$, is confined to $\mathcal{M},$ and we express the dynamics in the corresponding coordinate system. %tangent to it. %propose a new exponential submanifold design to give an exact solution for the quantum filter equation (\ref{1}) using differential geometric methods. 
In the following, we assume that $L$ is Hermitian, that is $L=L^\dag,$ and that $[H,L]=0,$ which corresponds to quantum non-demolition measurements \cite{braginsky1995quantum}.  In this case, we can write $L=\sum_{k=1}^{K} \lambda_{k} P_{k}$ and  $H=\sum_{j=1}^{D} \beta_{j} Q_{j}$, where the Hermitian operators $P_{k},Q_{j}$ are orthogonal projectors, that is $P_{k}P_{l}=\delta_{kl}P_{l}$ and $Q_{k}Q_{l}=\delta_{kl}Q_{l},$ satisfying $[P_k,Q_j]=0$ for every $k,j,$ and $K$ and $D$ are positive integers. Without loss of generality, we assume $K<n$ and $D<n.$ This is justified by the fact that replacing $L$ and $H$ by $L-\lambda_K\mathbb{I}$ and $H-\beta_D\mathbb{I}$ respectively, does not affect the normalized dynamics given by \eqref{1}.
\iffalse
 We reconsider the unnormalized Stratonovich form of (\ref{1}), given by \eqref{stra}.
 \begin{equation}
d \tilde{\rho}_{t}=\left(-i\left[H, \tilde{\rho}_{t}\right]-G \left(\tilde{\rho}_{t}\right)\right) d t+\sqrt{\eta}\left(L \tilde{\rho}_{t}+\tilde{\rho}_{t} L\right) \circ d Y_{t} \label{a}
\end{equation}
where $G\left(\tilde{\rho}_{t}\right)=(\eta-1)(L \tilde{\rho}_{t} L)+\frac{(1+\eta)}{2}(L^{2}\tilde{\rho}_{t}+\tilde{\rho}_{t} L^{2}).$\\
\fi

Let $\alpha=(\alpha_{11},\alpha_{12},\dots,\alpha_{K-1,K},\alpha_{KK})\in\mathbb R^{\frac{K(K+1)}{2}},$ $\theta\in\mathbb R^K,$ $\gamma\in\mathbb R^D,$  and  $\phi=(\theta,\gamma,\alpha)\in\mathbb{R}^N$, with $N:=K+D+\frac{K(K+1)}{2}$. Now define $\tilde{\rho}_{\phi}:=e^{\frac{1}{2}L_{\theta}+\frac{i}{2}H_{\gamma}}\rho_{\alpha}e^{\frac{1}{2}L_{\theta}-\frac{i}{2}H_{\gamma}}$ with 
\[L_{\theta}=\sum_{k=1}^{K} \theta_{k} P_{k},\quad H_{\gamma}=\sum_{j=1}^{D} \gamma_{j} Q_{j},\] 
and 
\[\rho_{\alpha}=\rho_{0}+\sum_{1\leq k\leq j\leq K}\left(P_{k}\rho_{0}P_{j}+(1-\delta_{kj})P_{j}\rho_{0}P_{k}\right)\left(e^{\alpha_{kj}}-1\right).\] 
It can be easily verified that $\tilde{\rho}_{\phi}\in\mathcal A.$ For the sake of simplicity, we assume that the set $\{\frac{\partial{\tilde{\rho}_{\phi}}}{\partial \phi_{1}}, \dots, \frac{\partial{\tilde{\rho}_{\phi}}}{\partial \phi_{N}}\}$
 is linearly
independent. Then $\mathcal {M}:= \{{\tilde{\rho}_{\phi}}\mid\phi\in\mathbb R^N\}$ is locally a $N$-dimensional differential submanifold of $\mathcal{A}$, with  tangent space  given by
\begin{equation}
T_{\tilde{\rho}_{\phi}}{\mathcal{M}}=span\left\{\frac{\partial{\tilde{\rho}_{\phi}}}{\partial \phi_{1}},\dots,\frac{\partial{\tilde{\rho}_{\phi}}}{\partial \phi_{N}}\right\}.
\end{equation}
A direct calculation yields
\begin{equation}
\frac{\partial{\tilde{\rho}_{\phi}}}{\partial \theta_{k}}=\frac{1}{2}(P_{k}\tilde{\rho}_{\phi}+\tilde{\rho}_{\phi}P_{k}),\label{eq-1}
\end{equation}
\begin{equation}
\frac{\partial{\tilde{\rho}_{\phi}}}{\partial \gamma_{j}}=\frac{i}{2}(Q_{j}\tilde{\rho}_{\phi}-\tilde{\rho}_{\phi}Q_{j}),\label{eq-2}
\end{equation}
\begin{equation}
\frac{\partial{\tilde{\rho}_{\phi}}}{\partial \alpha_{kj}}\!\!=\!e^{\frac{1}{2}L_{\theta}+\frac{i}{2}H_{\gamma}}\!(P_{k}\rho_{0}P_{j}+(1\!-\!\delta_{kj})P_{j}\rho_{0}P_{k})e^{\alpha_{kj}}e^{\frac{1}{2}L_{\theta}-\frac{i}{2}H_{\gamma}}.\label{eq-3}
\end{equation}
We have the following lemma,  which follows by direct calculation and by using \eqref{eq-1}, \eqref{eq-2} and \eqref{eq-3}.
\begin{lemma}
The terms $i\left[H, \tilde{\rho}_{\phi}\right]$, $F \left(\tilde{\rho}_{\phi}\right)$ and $L \tilde{\rho}_{\phi}+\tilde{\rho}_{\phi} L$ appearing in \eqref{stra} belong to the tangent space
$T_{\tilde{\rho}_{\phi}}{\mathcal{M}}.$
Furthermore,
\begin{align*}
i\left[H, \tilde{\rho}_{\phi}\right]&=2\sum_{j} \beta_{j}\frac{\partial{\tilde{\rho}_{\phi}}}{\partial \gamma_{j}},\\
F \left(\tilde{\rho}_{\phi}\right)&=(1-\eta)\sum_{k,j} \lambda_{k}\lambda_{j}\frac{\partial{\tilde{\rho}_{\phi}}}{\partial \alpha_{kj}}-(1+\eta)\sum_{k} \lambda_{k}^{2}\frac{\partial{\tilde{\rho}_{\phi}}}{\partial \theta_{k}},\\
L \tilde{\rho}_{\phi}+&\tilde{\rho}_{\phi} L=2\sum_{k} \lambda_{k}\frac{\partial{\tilde{\rho}_{\phi}}}{\partial \theta_{k}}.
\end{align*}
\label{lemma-tangent}
\end{lemma}		
%\begin{proof}
% {\textcolor{blue}{A direct calculation gives}}
%\begin{align*}
%i\left[H, \tilde{\rho}_{\phi}\right]&=i\sum_{j} \beta_{j}(Q_{j}\tilde{\rho}_{\phi}-\tilde{\rho}_{\phi}Q_{j})=2\sum_{j} \beta_{j}\frac{\partial{\tilde{\rho}_{\phi}}}{\partial \gamma_{j}},\\
%F \left(\tilde{\rho}_{\phi}\right)&=(1-\eta)(L \tilde{\rho}_{\phi} L)-\frac{(1+\eta)}{2}(L^{2}\tilde{\rho}_{\phi}+\tilde{\rho}_{\phi} L^{2})\nonumber\\&=(1-\eta)\sum_{k,j} \lambda_{k}\lambda_{j}\frac{\partial{\tilde{\rho}_{\phi}}}{\partial \alpha_{kj}}-(1+\eta)\sum_{k} \lambda_{k}^{2}\frac{\partial{\tilde{\rho}_{\phi}}}{\partial \theta_{k}},\\
%L \tilde{\rho}_{\phi}+&\tilde{\rho}_{\phi} L=\sum_{k} \lambda_{k}(P_{k}\tilde{\rho}_{\phi}+\tilde{\rho}_{\phi}P_{k})=2\sum_{k} \lambda_{k}\frac{\partial{\tilde{\rho}_{\phi}}}{\partial \theta_{k}}.
%\end{align*}
%\small
%The proof is complete.
%\end{proof}

%\vspace{0.3mm}
Now, we can establish the main result of this section.
\begin{theorem}
The solution $\tilde{\rho}_t$ of the quantum filter equation (\ref{1}) with initial condition $\rho_0$ coincides with $\tilde{\rho}_{\phi(t)}/{\rm Tr}(\tilde{\rho}_{\phi(t)}),$ where   $\phi(t)=(\theta(t),\gamma(t),\alpha(t)),$ with $\theta(t)$ satisfying  the stochastic differential equation
$$
d\theta_{k}(t)=-(1+\eta)\lambda_{k}^{2}dt+2\sqrt{\eta}\lambda_{k}dY_{t},\quad \theta_{k}(0)=0,
$$
and $\gamma_{j}(t)=-2\beta_{j}t,$ 
$\alpha_{kj}(t)=(1-\eta)\lambda_{k}\lambda_{j}t.$ 
\label{main}
\end{theorem}
\begin{proof}
By the previous lemma, the solutions of~\eqref{stra} evolve (almost surely) on $\mathcal{M}$ and satisfy %can be rewritten as 
\begin{align}
&d \tilde{\rho}_{\phi}=-2\sum_{j} \beta_{j}\frac{\partial{\tilde{\rho}_{\phi}}}{\partial \gamma_{j}}dt+(1-\eta)\sum_{k,j} \lambda_{k}\lambda_{j}\frac{\partial{\tilde{\rho}_{\phi}}}{\partial \alpha_{kj}}dt\nonumber\\&\!\!\!-(1+\eta)\sum_{k} \lambda_{k}^{2}\frac{\partial{\tilde{\rho}_{\phi}}}{\partial \theta_{k}}dt+2\sqrt{\eta}\sum_{k} \lambda_{k}\frac{\partial{\tilde{\rho}_{\phi}}}{\partial \theta_{k}}\circ dY_{t}. \label{c}
\end{align}
On other hand, by the chain rule we have
\begin{equation}
d \tilde{\rho}_{\phi}=\sum_{k}\frac{\partial{\tilde{\rho}_{\phi}}}{\partial \theta_{k}}\circ d\theta_{k}+\sum_{j}\frac{\partial{\tilde{\rho}_{\phi}}}{\partial \gamma_{j}}\circ d\gamma_{j}+\sum_{k,j}\frac{\partial{\tilde{\rho}_{\phi}}}{\partial \alpha_{kj}}\circ d\alpha_{kj}.
\end{equation}
To conclude, it is sufficient to identify the coefficients of the above equations with respect to the tangent space basis and solve the ordinary differential equations obtained for $\gamma_j$ and $\alpha_{kj}.$
\end{proof}
% {\textcolor{blue}{{\bf{Remark:}} When $\eta=1$, we get a reduced $N$-dimensional submanifold with $N = K + D$.}}
\section{Quantum projection filter and error analysis}\label{Q}
\subsection{ The projection filter approach}
The computation of the exact solution presented in Section~\ref{exact} is valid under the assumption of quantum non-demolition measurements. In this section, we follow an approach called projection filter, see, e.g.,~\cite{Amari,Vanhandel1}, which does not require the latter assumption and allows to further reduce the dimension of the system under study. This approach is mainly based on choosing an appropriate submanifold and suitably projecting the dynamics on it. %making a projection over the tangent space of that manifold.} %Here we consider an exponential manifold which is consider in \cite{Gao}  }}

To formalize this approach, let us introduce some quantum information geometry tools, mainly borrowed from \cite{Gao,Amari}. 

Recall that  $T_{\rho}\mathcal{Q}$, the tangent space of $\mathcal{Q}$ at the point $\rho,$ may be identified with $\mathcal{A}.$
When a tangent vector $X\in T_{\rho}\mathcal{Q}$ is considered as an element of $\mathcal{A}$ by this identification, we denote it by $X^{(m)}$ and we call it the $m$-representation of $X$.

%An $m$-representation of a tangent vector $X\in T_{\rho}\mathcal{Q}$, denoted by $X^{(m)}$, is the coordinate representation of $X$ with respect to the basis vectors
%\begin{equation}
%\partial_{i}^{(m)}=\partial_{i}:=\frac{\partial}{\partial \epsilon_i},
%\end{equation}
%where $(\epsilon_1,\dots,\epsilon_{n^2})$ are given coordinates on $\mathcal{Q}$. {\bf OK???}

%and given by a tangent vector $X$ $\in$ $T_{\rho}\mathcal{Q}$. We denote by $\rho_{\epsilon}$ the representation of a state $\rho$ in a given coordinate system {$\epsilon^{i}$}  on $\mathcal{Q}$. %Each state is  parameterized by $\rho_{\epsilon}$, if a coordinate system {$\epsilon^{i}$} is given on $\mathcal{Q}$. Also,
%We have
%\begin{equation}
%\partial_{i}^{(m)}=\partial_{i},
%\end{equation}
%is the $m$-representation of the natural basis vectors in the tangent space, where 
%$\partial_{i} := \frac{\partial{\rho_{\epsilon}}}{\partial{\epsilon_{i}}}$ and \{$\partial_{i}$\} are linearly independent and $T_{\rho}\mathcal{Q}=$Span\{${\partial_{i}}$\}.

%As a differential manifold does not have by nature an inner product structure, 
We define a symmetrized inner product $\langle\langle, \rangle\rangle_{\rho}$ on $T_{\rho}\mathcal{Q}\equiv\mathcal{A}$ as follows:
\begin{equation}
\label{eq:inner}
\langle\langle A, B\rangle\rangle_{\rho}=%\frac{1}{2}
\frac12{\rm Tr}(\rho AB + \rho BA),\quad \forall A, B \in \mathcal{A}.
\end{equation}
Next, the $e$-representation of a tangent vector $X \in {T_{\rho}\mathcal{Q}}$ is defined as
the Hermitian operator $X^{(e)} \in \mathcal{A}$ satisfying
\begin{equation}
\label{eq:innerem}
\langle\langle X^{(e)}, A\rangle\rangle_{\rho}={\rm Tr}(X^{(m)}A),\quad \forall A \in \mathcal{A}.
\end{equation}
By using~\eqref{eq:inner} and~\eqref{eq:innerem} it is easy to obtain
\begin{equation}
\label{eq:emrepr}
X^{(m)} = \frac{1}{2}(\rho X^{(e)} +X^{(e)} \rho),\quad \forall X\in T_{\rho}\mathcal{Q}.
\end{equation}
Using the $e$-representation defined above, a further inner
product on $T_{\rho}\mathcal{Q}$ is defined by
\[
\langle X,Y \rangle_{\rho}\!=\!\langle\langle X^{(e)},Y^{(e)}\rangle\rangle_{\rho}\!=\!{\rm Tr}(X^{(m)}Y^{(e)}), \quad \forall X, Y \in T_{\rho}\mathcal{Q}.    
\]
The quantum Fisher metric is a Riemannian metric $g$ whose components are
\begin{equation}
\label{eq:Riem}
g_{ij} =\langle \partial_{i},\partial_{j} \rangle_{\rho}={\rm Tr}(\partial_{i}^{(m)}\partial_{j}^{(e)}),
\end{equation}
where $\partial_i:=\frac{\partial}{\partial \epsilon_i}$ and $(\epsilon_1,\dots,\epsilon_{n^2})$ are given coordinates on $\mathcal{Q}$.
%More details about this information geometric tools can be found in \cite{1} and \cite{11}.

Following~\cite{Brigo2,Gao} (in the classical and quantum framework, respectively), we consider the subset %a $m$-dimensional differential submanifold 
$\mathscr{S}$ of $\mathcal{Q}$ consisting of
an exponential family of unnormalized quantum density operators
\[\mathscr{S}= \{{\tilde{\rho}_{\theta}}\mid \theta = (\theta_{1}, \dots, \theta_{m}) \in \Theta \}.\]
Here $\tilde{\rho}_{\theta}:=e^{\frac{1}{2}\sum_{i=1}^{m} \theta_{i} A_{i}}\rho_{0}e^{\frac{1}{2}\sum_{i=1}^{m} \theta_{i} A_{i}}$, $\rho_0$ is the initial condition for the (projected) dynamics, the  operators $A_{i} \in \mathcal{A}$, for $i \in \{1, 2,\dots, m\}$, are
 assumed to be mutually commuting and pre-designed, and $\Theta$ is an open subset of $\mathbb{R}^{m}$
containing the origin. 
%We assume that the entire submanifold $\mathcal{M}$ can be covered by a single coordinate chart $(\mathcal{M}, \theta = (\theta_{1}, . . . , \theta_{m}) \in \Theta)$, where $\Theta$ is an open subset of $\mathbb{R}^{m}$ containing the origin, see \cite{Gao}.\\
%Due to the chain rule in Stratonovich stochastic calculus, we have \begin{equation}
%d \tilde{\rho}_{\theta}=\sum_{i}\frac{\partial{\tilde{\rho}_{\theta}}}{\partial \theta_{i}}\circ d\theta_{i}. \label{5}
%\end{equation}
 Assuming that the set \{$\frac{\partial{\tilde{\rho}_{\theta}}}{\partial \theta_{1}}, \dots, \frac{\partial{\tilde{\rho}_{\theta}}}{\partial \theta_{m}}$\}
 is linearly
independent, we obtain that $\mathscr{S}$ is, locally, a $m$-dimensional differential submanifold of $\mathcal{Q}$.
The tangent  space at some $\tilde{\rho}_{\theta}\in \mathscr{S}$ is given by 
$
T_{\tilde{\rho}_{\theta}}\mathscr{S} = {\rm span}\{{\tilde{\partial_{i}}, i = 1,\dots, m}\},$ \vspace{2mm}
where $\tilde{\partial_{i}}^{(m)}=\frac{\partial{\tilde{\rho}_{\theta}}}{\partial \theta_{i}}=\frac{1}{2}(A_{i}\tilde{\rho}_{\theta}+\tilde{\rho}_{\theta}A_{i}).$
Using~\eqref{eq:emrepr} we get 
$
\tilde{\partial_{i}}^{(e)} = A_{i}.
$
In analogy with~\eqref{eq:Riem} we define a Riemannian metric on $\mathscr{S}$ whose components are  real-valued functions of $\theta$:
\begin{align}
g_{ij}(\theta)={\rm Tr}(\tilde{\partial_{i}}^{(m)}\tilde{\partial_{j}}^{(e)})={\rm Tr}(\tilde{\rho}_{\theta}A_{i}A_{j}).
\end{align}
The matrix $G(\theta):=(g_{ij}(\theta))_{i,j=1,\dots,m}$ is a quantum Fisher information matrix. Then, for every $\theta  \in \Theta$, we can define an orthogonal projection
operation $\Pi_{\theta}$ by
\begin{align}
\Pi_{\theta}:\mathcal{A} &\longrightarrow T_{\tilde{\rho}_{\theta}}\mathscr{S} \nonumber\\
x &\longmapsto \sum_{i=1}^{m}\sum_{j=1}^{m}g^{ij}(\theta)\langle\langle x,\tilde{\partial_{j}}^{(e)} \rangle\rangle_{\tilde{\rho}_{\theta}}\tilde{\partial_{i}}  \label{6},
\end{align}
where the $g^{ij}(\theta)$ are the components of the inverse of the quantum information
matrix $G(\theta)$.\\
%Consider a curve $\varphi : t \longmapsto \rho_{t}$ in $\mathscr{S}$. Through the (local) coordinate chart ($\mathscr{S}, \theta$), this curve corresponds to a curve $\gamma : t \longmapsto \theta_{t}$ in $\Theta$.
%Assuming that the curve $\varphi$ starts from the initial condition $\tilde{\rho}_{\theta_{0}} = \tilde{\rho}_{0}$, or equivalently, the curve $\gamma$ starts from $\theta_{0} = 0$, 
We define the quantum projection filter on $\mathscr{S}$ by
\begin{align}
d \tilde{\rho}_{\theta_{t}}=\Pi_{\theta_{t}}\left(-i\left[H, \tilde{\rho}_{\theta_{t}}\right]\right)d t+\Pi_{\theta_{t}}\left(-F \left(\tilde{\rho}_{\theta_{t}}\right)\right) d t\nonumber\\+\Pi_{\theta_{t}}\left(\sqrt{\eta}(L \tilde{\rho}_{\theta_{t}}+\tilde{\rho}_{\theta_{t}} L^{\dagger})\right) \circ d Y_{t}. \label{7}
\end{align}
%This projection can be summarized in two steps. The first step is to evaluate the SDE \eqref{3} at the state $\tilde{\rho}_{\theta_{t}}$, instead of the true state $\rho_{t}$. The second step is projecting the vector fields obtained at the state $\tilde{\rho}_{\theta_{t}}$ into the tangent vector space $T_{\tilde{\rho}_{\theta}}\mathscr{S}$.
Since the vector fields regulating the dynamics are everywhere tangent to $\mathscr{S}$,  the solution of the previous equation is a well-defined stochastic process $\tilde{\rho}_{\theta_{t}}$ on $\mathscr{S},$ whenever $\tilde{\rho}_{\theta_0} = {\rho}_0$ belongs to $\mathscr{S}.$
Similarly to \cite{Gao}, by using the orthogonal projection operation and the chain rule
\begin{equation}
d \tilde{\rho}_{\theta}=\sum_{i}\frac{\partial{\tilde{\rho}_{\theta}}}{\partial \theta_{i}}\circ d\theta_{i}, \label{5}
\end{equation}
we can easily express the dynamics of the parameter $\theta_t$ as %the quantum projection filter:
\begin{equation}
d \theta_{t}=G\left(\theta_{t}\right)^{-1}\left\{\Xi\left(\theta_{t}\right) d t+\Gamma\left(\theta_{t}\right) \circ d Y_t\right\} \label{8}
\end{equation}
with $\theta_{i}(0)=0,$ for $i=1, \ldots, m$. Here, the $j$-th elements of the $m$-dimensional column vectors $\Xi\left(\theta_{t}\right)$ and $\Gamma\left(\theta_{t}\right)$ are 
\begin{equation*}
\Xi_{j}\left(\theta_{t}\right)=\operatorname{Tr}\left(\tilde{\rho}_{\theta_{t}}\left(i\left[H, A_{j}\right]-F^{\dagger}\left(A_{j}\right)\right)\right)
\end{equation*}
and 
%\begin{equation*}
$\Gamma_{j}\left(\theta_{t}\right)=\sqrt{\eta}\operatorname{Tr}\left(\tilde{\rho}_{\theta_{t}}\left(A_{j} L+L^{\dagger} A_{j}\right)\right).$
%\end{equation*}
Let $\rho_{\theta_{t}}=\frac{\tilde{\rho}_{\theta_{t}}}{{\rm Tr}(\tilde{\rho}_{\theta_{t}})}$ be the normalized approximate quantum information state.
We note that only $m$ SDEs need to be solved for $\rho_{\theta_{t}}$ instead of $n^{2}-1$ for the original quantum filter.

\subsection{Error analysis}
Following \cite{Brigo2}, we define at each point $\tilde{\rho}_{\theta_{t}}$ the prediction residual as
%%{
%%\small
%\begin{equation*}
$\Omega(t)= -i\left[H, \tilde{\rho}_{\theta_{t}}\right]-\Pi_{\theta_{t}}\left(-i\left[H, \tilde{\rho}_{\theta_{t}}\right]\right)$
%\end{equation*}
%%}
and the two correction residuals as
%{
%\small
\begin{align*}
C_{1}(t)&=-F(\tilde{\rho}_{\theta_{t}})-\Pi_{\theta_{t}}\left(-F(\tilde{\rho}_{\theta_{t}})\right),\\
%\end{equation*} 
%}
%\vspace{-4mm}
%{
%%\small
%\begin{equation*}
C_{2}(t)&=\sqrt{\eta}\left(L \tilde{\rho}_{\theta_{t}}+\tilde{\rho}_{\theta_{t}} L^{\dagger}\right)-\Pi_{\theta_{t}}\left(\sqrt{\eta}(L \tilde{\rho}_{\theta_{t}}+\tilde{\rho}_{\theta_{t}} L^{\dagger})\right),
\end{align*}
%}
respectively. These residuals refer to the local approximation errors due to the projection of the vector fields $ -i\left[H, \tilde{\rho}_{\theta_{t}}\right]$, $-F(\tilde{\rho}_{\theta_{t}})$ and $\sqrt{\eta}\left(L \tilde{\rho}_{\theta_{t}}+\tilde{\rho}_{\theta_{t}} L^{\dagger}\right)$ into the tangent space $T_{\tilde{\rho}_{\theta}}\mathscr{S}$ at time $t$.
For the sake of simplicity, we assume that the operator $L$ is Hermitian. 
This assumption simplifies the analysis of the local errors. %This is verified in many physical setups. %This is also advantageous in many experimental settings; such as trapping a cold atomic ensemble in an optical
%cavity \cite{Thomsen,Vanhandel2}. 
By using the spectral theorem, we can write $
L=\sum_{i=1}^{n_{0}} \lambda_{i} P_{{i}},
$ where $n_{0} \leq n$ is the number of nonzero distinct eigenvalues
of $L$ denoted by {$\lambda_{i},$}  and
$P_{{i}}$ are orthogonal projections.  %{\textcolor{blue}{satisfying $
%Now \eqref{3} can be written as 
%\begin{equation}
%d \tilde{\rho}_{t}=\left(-i\left[H, \tilde{\rho}_{t}\right]-G \left(\tilde{\rho}_{t}\right)\right) d t+\sqrt{\eta}\left(L \tilde{\rho}_{t}+\tilde{\rho}_{t} L\right) \circ d Y_{t}
%\end{equation}
%where $G\left(\tilde{\rho}_{t}\right)=(\eta-1)(L \tilde{\rho}_{t} L)+\frac{(1+\eta)}{2}(L^{2}\tilde{\rho}_{t}+\tilde{\rho}_{t} L^{2}) .$\\

 Let us set $M(\tilde{\rho}_{\theta_{t}}):=A_{k} \tilde{\rho}_{\theta_{t}}+\tilde{\rho}_{\theta_{t}}A_{k}$ and $X_{0}:=-i[H,\rho_{0}].$  We have the following result. %Now we state the following proposition.
\begin{proposition}\label{B}
Assume $m=n_{0}$ and  
$A_{i}=P_{{i}}.$ Then, the correction residuals are $$C_{1}(t)=\sum_{k=1}^{m}(\eta -1) \lambda_{k}^{2}\Big(\frac{1}{2}(A_{k} \tilde{\rho}_{\theta_{t}}+ \tilde{\rho}_{\theta_{t}} A_{k}) - A_{k} \tilde{\rho}_{\theta_{t}} A_{k}\Big)$$ and $C_{2}(t)=0$ $\forall t \geq 0.$ 

Moreover, if $[H,L]=0$, %\footnote{This assumption is usually referred to as QND measurement.}, 
then the exponential quantum projection filter equation \eqref{8} becomes
\begin{equation}
d \theta_i(t)=-2 \eta \lambda_i^2 dt+2 \sqrt{\eta} \lambda_i d Y_t,\quad i=1,\dots,m \label{14}
\end{equation}
%{\textcolor{blue}{where $\alpha=(\lambda_{1}^{2},\dots,\lambda_{m}^{2})^{\prime}$ and $\beta=(\lambda_{1},\ldots,\lambda_{m})^{\prime}$,}}\\
%and the correction residuals $C_{1}(t)$ and $C_{2}(t)$ remain as in Theorem \ref{B} for all $t\geq 0$.\\
and the prediction residual $\Omega(t)$ is given by
\begin{equation}
\Omega(t)=e^{\frac{1}{2}\sum_{i=1}^{m} \theta_{i}(t) A_{i}}X_{0}e^{\frac{1}{2}\sum_{i=1}^{m} \theta_{i}(t) A_{i}}, t \geq 0. \label{15}
\end{equation}
%In this case, the exponential quantum projection filter \eqref{8} is simplified to
%\begin{equation}
%d \theta_{t}=G(\theta_{t})^{-1}{\rm Tr}(i\tilde{\rho}_{\theta_{t}}[H, A_{j}])dt -2 \eta \alpha dt+2 \sqrt{\eta} \beta  d Y_t \label{11}
%\end{equation}
%where $\alpha=(\lambda_{1}^{2},\ldots,\lambda_{m}^{2})$ and $\beta=(\lambda_{1},\ldots,\lambda_{m})'$.
\end{proposition}
\begin{proof}
By %using the %
definitions of $C_{1}(t)$ and $C_{2}(t)$, one has
\begin{align*}
C_{1}(t)&=\Pi_{\theta_{t}}\left(F(\tilde{\rho}_{\theta_{t}})\right)-F(\tilde{\rho}_{\theta_{t}}) \nonumber\\
%&=\sum_{k,l=1}^{m} (\eta -1) \lambda_{k}\lambda_{l} \bigg(\Pi_{\theta_{t}}\left(A_{k} \tilde{\rho}_{\theta_{t}}A_{l}\right)-\left(A_{k} \tilde{\rho}_{\theta_{t}}A_{l}\right)\bigg) \nonumber\\%
%&+\sum_{k=1}^{m} \frac{\eta +1}{2} \lambda_{k}^{2}\bigg( \Pi_{\theta_{t}}\left(M(\tilde{\rho}_{\theta_{t}})\right)-M(\tilde{\rho}_{\theta_{t}})\bigg) \nonumber\\%
%&=\sum_{k=1}^{m} (\eta -1) \lambda_{k}^{2}\bigg( \frac{1}{2}(A_{k} \tilde{\rho}_{\theta_{t}}+\tilde{\rho}_{\theta_{t}}A_{k})-A_{k} \tilde{\rho}_{\theta_{t}}A_{k}\bigg).%
&=\sum_{k=1}^{m}(\eta -1) \lambda_{k}^{2}\bigg(\frac{1}{2}\left(M(\tilde{\rho}_{\theta_{t}})\right) - A_{k} \tilde{\rho}_{\theta_{t}} A_{k}\bigg),\\
%\end{align}
%
%\vspace{-4mm}
%
%%\small
%\begin{align}
C_{2}(t)&=\Pi_{\theta_{t}}\left(\sqrt{\eta}(L \tilde{\rho}_{\theta_{t}}+\tilde{\rho}_{\theta_{t}} L)\right)-\left(\sqrt{\eta}(L \tilde{\rho}_{\theta_{t}}+\tilde{\rho}_{\theta_{t}} L)\right) \nonumber\\
&=\textstyle{\sum_{k=1}^{m} 2 \sqrt{\eta} \lambda_{k}\big(\Pi_{\theta_{t}}(\tilde{\partial}_{k})-\tilde{\partial}_{k}\big)} =0.
\end{align*}
Equations~\eqref{14}-\eqref{15} follow from the fact that $[H,L]=0$ implies $\Xi_{j}(\theta_{t})=-2 \eta \lambda_{j}^{2}{\rm Tr}(\tilde{\rho}_{\theta_{t}}A_{j})$, $
\Gamma_{j}\left(\theta_{t}\right)=2\sqrt{\eta}\lambda_{j}{\rm Tr}(\tilde{\rho}_{\theta_{t}}A_{j})
$ and $g^{ij}(\theta) = \delta_{ij}{\rm Tr}(\tilde{\rho}_{\theta_{t}}A_{j}).$
%The rest of the proof is a direct consequence of the following identities 
%$$
%\Xi_{j}\left(\theta_{t}\right)={\rm Tr}(i\tilde{\rho}_{\theta_{t}}[H, A_{j}])-2 \eta \lambda_{j}^{2}{\rm Tr}(\tilde{\rho}_{\theta_{t}}A_{j}),
%$$  and $
%\Gamma_{j}\left(\theta_{t}\right)=2\sqrt{\eta}\lambda_{j}{\rm Tr}(\tilde{\rho}_{\theta_{t}}A_{j}).
%$
%To conclude, it is sufficient to apply the fact that $[H,L]=0.$
\end{proof}

%\begin{proof}
%Since $[H, L] = 0$ and $A_{i} = P_{{i}}$
%is the projection operator
%of $L$, one has $[H, A_{i}] = 0, i = 1, 2, . . . , m.$ Then we get the It$\hat{o}$ stochastic differential equation 
%\eqref{14}.
%\end{proof}
\vspace{0.8mm}
Let $\mathbb P$ denote the original probability measure under which $W_t$ is a Wiener process. By Girsanov theorem, there exists an equivalent probability measure $\mathbb P^{\prime}$ such that $Y_t$ in \eqref{eq:y} becomes a Wiener process. Let $\mathbb{E}$ denote the expectation with respect to the measure $\mathbb{P^{\prime}}.$ 

To measure the gap between %difference of the variations of 
the filter state and its approximation, we consider the average total residual norm defined as
\begin{equation}
e_{t}:=\mathbb{E}\left\|C_{1}(t)+C_{2}(t)+\Omega(t)\right\|_{F},
\end{equation}
with $e_{0}=0$. Also, set $Y_{k}:=\frac{1}{2}(A_{k}\rho_{0}+\rho_{0}A_{k})-A_{k}\rho_{0}A_{k},$
and $\sigma:=(1-\eta)\max_{k} \lambda_{k}^{2}.$
We now state the main result of this section.
\begin{theorem} \label{C}
Let the assumptions of Proposition \ref{B} hold true. If $[H,L]=0,$ then 
{\small
\begin{align}
e_{t} &\leq \sigma\sqrt{\sum_{k=1}^{m}{\rm Tr}(Y_{k}^{2})+\sum_{j \neq j^{\prime}}^{m}\left(s_{1}(Y_{j})\sum_{i=1}^{m}s_{i}(Y_{j^{\prime}})\right)}
+\sqrt{{\rm Tr}(X_{0}^{2})}. \label{16}
\end{align}
}
\end{theorem}
\begin{proof}
Let us firstly note that  $e_{t}=\mathbb{E}\left\|C_{1}(t)+\Omega(t)\right\|_{F}$.	By using the triangular inequality, we get
$$e_t \leq \mathbb{E}\left\|C_{1}(t)\right\|_{F}+\mathbb{E}\left\|\Omega(t)\right\|_{F}.$$
Now, let $\Delta(t):=\frac{1}{2}\sum_{i=1}^{m} \theta_{i}(t) A_{i}.$ We have
\begin{align}
&C_{1}(t)=\sum_{k=1}^{m}(\eta -1) \lambda_{k}^{2}\bigg(\frac{1}{2}\left(M(\tilde{\rho}_{\theta_{t}})\right) - A_{k} \tilde{\rho}_{\theta_{t}} A_{k}\bigg)\nonumber\\
%&=\sum_{k=1}^{m}(\eta -1) \lambda_{k}^{2}\bigg(e^{\Delta(t)}\left(\frac{1}{2}\left(M(\rho_{0})\right ) -  A_{k}\rho_{0} A_{k} \right)e^{\Delta(t)}\bigg)\nonumber\\%
&=\sum_{k=1}^{m}(\eta -1) \lambda_{k}^{2}\bigg(e^{\Delta(t)}Y_{k}e^{\Delta(t)}\bigg).
\end{align}
Define $Z_{k}:=e^{2\Delta(t)}Y_{k}$  and $\tilde Z_{k}:=\sum_{i=1}^m s_i(Y_k)$. By using Lemma A.1 we get 
{\small
\begin{align}
&\mathbb{E}\left\|C_{1}(t)\right\|_{F}=\mathbb{E}\sqrt{{\rm Tr}(C_{1}(t)^{2}})\nonumber\\
%5&\!\!\!\leq\mathbb{E}\sqrt{\sigma^2 Tr\left(\left(\sum_{k=1}^{m}e^{\Delta(t)}Y_{k}e^{\Delta(t)} \right)^{2}\right)}\nonumber\\%
&\!\!\!= \sigma \mathbb{E}\sqrt{\sum_{k=1}^{m}Tr\left(Z_{k}^{2}\right)+\sum_{j \neq j^{\prime}}^{m}Tr\left( Z_{j}Z_{j^{\prime}}\right)} \nonumber\\
%&\!\!\!\leq \sigma \mathbb{E}\sqrt{\sum_{k=1}^{m}\left(\sum_{i=1}^{m}s_{i}\left(Z_{k}^{2}\right)\right)+\sum_{j \neq j^{\prime}}^{m}\left(\sum_{i=1}^{m}s_{i}\left(Z_{j}Z_{j^{\prime}}\right)\right)} \nonumber\\%
&\!\!\!\leq \sigma \mathbb{E}\sqrt{\sum_{k=1}^{m}\left(\sum_{i=1}^{m}s_{i}^{2}(Z_{k})\right)+\sum_{j \neq j^{\prime}}^{m}\left(\sum_{i=1}^{m}s_{i}(e^{2\Delta(t)})s_{i} (Y_{j}Z_{j^{\prime}})\right)} \nonumber\\
&\!\!\!\leq \sigma \mathbb{E}\sqrt{\sum_{k=1}^{m}\left(\sum_{i=1}^{m}s_{i}(e^{4\Delta(t)}Y_{k}^{2})\right)+\sum_{j \neq j^{\prime}}^{m}\left(s_{1}^{2}(e^{2\Delta(t)})s_{1}( Y_{j})\tilde Z_{j^{\prime}}\right)} \nonumber\\
&\!\!\!\leq \sigma \mathbb{E}\sqrt{s_{1}^{2}(e^{2\Delta(t)}) \left( \sum_{k=1}^{m}\Bigg(\sum_{i=1}^{m}s_{i}^{2}(Y_{k})\Bigg)+\sum_{j \neq j^{\prime}}^{m}\Bigg(s_{1}( Y_{j})\tilde Z_{j^{\prime}})\Bigg)\right)}  \nonumber\\
&\!\!\!\leq \sigma\sqrt{ \sum_{k=1}^{m}{\rm Tr}(Y_{k}^{2})+\sum_{j \neq j^{\prime}}^{m}\left(s_{1}( Y_{j})\tilde Z_{j^{\prime}}\right)}\max_{i}\mathbb{E}e^{\theta_{i}(t)}, \label{17}
%&\!\!\!=\sigma\sqrt{ \sum_{k=1}^{m}{\rm Tr}(Y_{k}^{2})+\sum_{j \neq j^{\prime}}^{m}\left(s_{1}( Y_{j})\tilde Z_{j^{\prime}}\right)}%
\end{align}
}

\noindent where $\max_{i}\mathbb{E}e^{\theta_{i}(t)}=\max_{i}e^{\theta_{i}(0)}=1.$ This comes from the fact that $e^{\theta_{i}(t)}$ is a martingale  with respect to $\mathbb P'.$

Similarly, we get
\begin{align}
&\mathbb{E}\left\|\Omega(t)\right\|_{F}=\mathbb{E}\sqrt{{\rm Tr}(\Omega(t)^{2}}) \leq \sqrt{{\rm Tr}(X_{0}^{2}}). \label{18}
\end{align}
Adding up \eqref{17} and \eqref{18}, we obtain the inequality \eqref{16}.
\end{proof}
Under some additional conditions, Theorem \ref{C} leads to an equivalence between the exponential quantum projection
filter equation \eqref{7} and the quantum filter equation \eqref{zakai}. 
\begin{Corollary}
Let the assumptions of Proposition \ref{B} hold true and assume in addition that $[H,L]=[H,\rho_0]=[L,\rho_0]=0.$ Then $\tilde{\rho}_{t} \equiv \tilde{\rho}_{\theta_{t}}.$
\label{errorexact}
\end{Corollary}

\subsection{Quantum state reduction}
Under the quantum non-demolition assumption $[H,L]=0$,  the normalized evolution of the quantum projection filter $\rho_{\theta_{t}}=\frac{\tilde{\rho}_{\theta_{t}}}{{\rm Tr}(\tilde{\rho}_{\theta_{t}})}$ can be written as
\begin{align}
d {\rho}_{\theta_{t}}&= \eta  \left(L\rho_{\theta_{t}}L-\frac{L^{2}\rho_{\theta_{t}}}{2}-\frac{\rho_{\theta_{t}}L^{2}}{2}\right) dt \nonumber\\
&+\sqrt{\eta } \left(L\rho_{\theta_{t}}+\rho_{\theta_{t}}L-2\operatorname{Tr}(L\rho_{\theta_{t}})\rho_{\theta_{t}}\right) d\hat{W}_{t}, \label{Behavior}
\end{align}
where $d\hat{W}_{t}=dY_{t}-2\sqrt{\eta }\operatorname{Tr}(L\rho_{\theta_{t}})dt.$
As in the previous section, let us write $
L=\sum_{i=1}^{n_{0}} \lambda_{i} P_{{i}},
$ where $n_{0} \leq n$ is the number of nonzero distinct eigenvalues
of $L$ denoted by {$\lambda_{i},$}  and
$P_{{i}}$ are orthogonal projections.
The following result states that the quantum state reduction phenomenon occurs for both the evolutions given by~\eqref{1} and by~\eqref{Behavior}; it can be obtained by following standard stochastic LaSalle-type arguments similarly to~\cite{Vanhandel3}, using the Lyapunov function $V(\rho)=\operatorname{Tr}(L^{2}\rho)-\operatorname{Tr}^{2}(L\rho)$.
\begin{theorem}
For every initial condition $\rho_0\in S$, the solution %$\rho_t$ and 
$\rho_{\theta_{t}}$ of %\eqref{1} and~
\eqref{Behavior} %for any $\rho_{0} \in S$, 
converge a.s. as $t \rightarrow \infty$ to one of the subsets $\{\rho\in S\mid P_k\rho=\rho\}$, for $k=1,\dots,n_0$. 
The same property holds true for the solution $\rho_t$ of~\eqref{1}.  %$\rho_{k}=e_{k}e_{k}^{*}$ %with $k \in \{0,\ldots,2J\}.$%
\end{theorem}
Note that the previous result shows that the solutions of~\eqref{1} and~\eqref{Behavior} share a similar asymptotic behavior, but it does not guarantee that such solutions converge almost surely to the same limit.
The results obtained in~\cite{liang2021robust,liang2022model,benoist2014large} suggest that such limits coincide.
%Assuming that the limits are the same, then 
It is then natural to expect that a feedback control depending on the quantum projection filter may be used to stabilize the system towards a chosen eigenstate of $L$, similarly to what was done in, e.g.,~\cite{Vanhandel3,liang2019exponential}.

\vspace{-3mm}
\section{Numerical simulations}\label{R}
\subsection{A spin-$\frac12$ system}
Here we present simulation results for the simple case of a spin-$\frac 1 2$ system. For a two-level quantum system, $\rho$ can be uniquely characterized by the Bloch sphere coordinates $(x, y, z)$ as $\rho=\frac{1}{2}\begin{psmallmatrix}
1+z & x-iy \\ x+iy & 1-z
\end{psmallmatrix}$.
The vector $(x, y, z)$ belongs to the ball
$B(\mathbb{R}^3):= \{(x, y, z) \in \mathbb{R}^3: x^2 + y^2 + z^2 \leq1\}.$
We take $H=\frac{\omega_{eg}}{2}\sigma_{z}$ and $L=\frac{\sqrt{M}}{2}\sigma_{z}$, where $w_{eg}$ and $M > 0$ are physical parameters. 

It can be verified that the dynamics in the Bloch sphere coordinates are given by 
\begin{equation}
\left\{\begin{array}{l}
d x_t=\left(-\frac{M}{2} x_t-\omega_{eg} y_t\right) d t- \sqrt{\eta M } x_t z_t d W(t) \\
d y_t=\left(\omega_{eg} x_t-\frac{M}{2} y_t\right) d t- \sqrt{\eta M} y_t z_t d W(t) \\
d z_t=\sqrt{\eta M}\left(1-z_t^2\right) d W(t)
\end{array}\right.
\label{bloch}
\end{equation}

%Since $L=\frac{\sqrt{M}}{2}\sigma_{z}$ is self-adjoint, by the spectral theorem, $L$ can be written as
The operator $L$ can be written as $L= \lambda_{1}P_{{1}}+\lambda_{2}P_{{2}},$ where $ \lambda_{1}=\frac{\sqrt{M}}{2}$ and $ \lambda_{2}=-\frac{\sqrt{M}}{2},$ $P_{{1}}=\begin{psmallmatrix}
1 & 0 \\ 0 & 0
\end{psmallmatrix},$ and $P_{{2}}=\begin{psmallmatrix}
0 & 0 \\ 0 & 1
\end{psmallmatrix}.$
We note that $dY(t) = dW(t)+\sqrt{\eta M}{\rm Tr}(\sigma_{z}\rho_{t})$ is used to drive the exponential
quantum projection filter. Here, the matrices
$\sigma_{x}$, $\sigma_{y}$, and $\sigma_{z}$ correspond to the Pauli matrices. We take $t \in [0,T]$ with $T=5$, and step size $\delta=2^{-12}T$. Also, 
$w_{eg}=1$, $\eta = 0.5$, $M=1$, $\alpha=7.61$, $\beta=5$, and $\gamma=10$. The initial state is $\rho_{0}=(-1,0,0)$.
Figure \ref{frob} shows the Frobenius norm of the difference between $\rho_t$ and $\rho_{\theta_t}$. %, while Figure \ref{prob} displays their components along the eigenstate $P_{L{_1}}.$
\begin{figure}
\centerline{\includegraphics[width=2.6in]{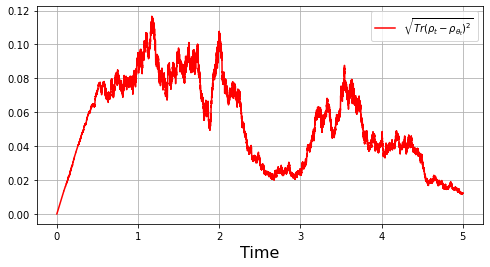}}
\caption{Approximation error between the quantum filter and the quantum projection filter.}
\label{frob}
\end{figure}
%\begin{figure}
%\centerline{\includegraphics[width=2.6in]{w}}
%\caption{Components of the quantum filter and the quantum projection filter along $P_{L{_1}}.$}
%\label{prob}
%\end{figure}
%In this section, we psimulation results from a spin-$\frac{1}{2}$
%system example is used to confirm our theoretical results. 
\subsection{Discussion on the error in the presence of a feedback}
Our goal is to study whether the  approach developed in the previous sections remains effective in the presence of a controlled Hamiltonian. In particular, we wonder whether the quantum projection filter  is a good candidate to replace the original filter  in  the   stabilizing feedback law introduced in \cite{Liang}. In that paper, the dynamics of a controlled spin-$\frac 12$ generalizing the dynamics \eqref{bloch} in presence of a control law $u_t$ takes  the following form 
$$\left\{\begin{array}{l}
d x_t=\left(-\frac{M}{2} x_t-\omega_{eg} y_t+u_tz_t\right) d t- \sqrt{\eta M } x_t z_t d W(t) \\
d y_t=\left(\omega_{eg} x_t-\frac{M}{2} y_t\right) d t- \sqrt{\eta M} y_t z_t d W(t) \\
d z_t=-u_t x_t d t+ \sqrt{\eta M}\left(1-z_t^2\right) d W(t)
\end{array}\right.$$
In \cite{Liang}, 
a feedback controller $u_t=u(\rho_t)$ is applied to stabilize the above system towards the excited state $\rho_e$ corresponding to the Bloch sphere coordinates $(0,0,-1).$  The feedback takes the form 

\vspace{-3pt}

\begin{equation}u(\rho)=\alpha[V(\rho)]^{\beta}- \gamma {\rm Tr}(i[\sigma_{y},\rho]\rho_{e}),
\label{feedback}
\end{equation}

\vspace{-3pt}

\noindent where $V (\rho) \!=\!\sqrt{1-{\rm Tr}(\rho\rho_{e})}$, with $\alpha>0,$ $\beta \geq 0,$ and $\gamma \geq1$.

Here we assume that the feedback law \eqref{feedback} is evaluated at $\rho_\theta$ instead of $\rho$ and we study numerically the stabilization towards the excited state. The simulation parameters are the same as before. 
% with $\alpha=7.61$, $\beta=5$, and $\gamma=10$.%
The validity of the proposed approximation filtering
scheme is checked through the Frobenius norm of the difference between $\rho_{t}$ and $\rho_{\theta_{t}}$ in Figure \ref{frobf}. %One can see that $\rho_{t}$ and $\rho_{\theta_{t}}$ are very close
%over this time interval. 
Figure \ref{frobft} shows the convergence towards the target state. %One can observe that
%$\rho_{t}$ converges to the target state.
\begin{figure}
\centerline{\includegraphics[width=2.6in]{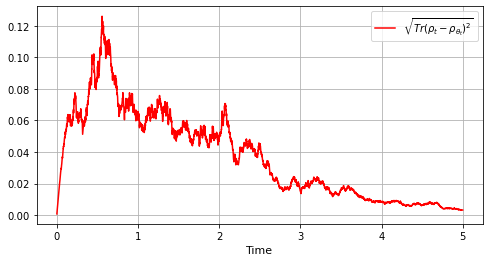}}
\caption{Approximation error between the quantum filter and the quantum projection filter in presence of a feedback control based on the projection filter.}
\label{frobf}
\end{figure}
%\vspace{-2.73mm}
\begin{figure}
\centerline{\includegraphics[width=2.6in]{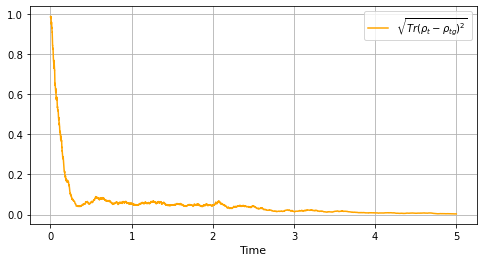}}
\caption{Convergence of the quantum filter to the target state $\rho_{e}$ by applying a feedback control based on the projection filter.
}
\label{frobft}
\end{figure}

\section{Conclusions and Future works}\label{S}
In this paper we first develop an approach allowing us to derive the exact solution of the filter equation under QND measurement with imperfect measurements. Such a solution is described in terms of a solution of a simplified stochastic differential equation. To further reduce the complexity of the dynamics, we generalize the projection filter approach developed in~\cite{Gao} to the case of imperfect measurements. An analysis of the approximation error has been performed and a quantum state reduction result for the projected dynamics has been shown in the case of QND measurement. Simulations of a two-level  system are provided with the aim of verifying the efficiency of the projection filtering method in the feedback stabilization design.
In future work, we aim at improving the error estimate, for instance by making use of the approach established in~\cite{gao2020improved}, where a projection filter design for the case of perfect measurements was provided based on Stratonovich stochastic Taylor expansions.
%OR
%{\color{green} In future work, we aim at improving the error estimate, by providing a new quantum projection filter design based on Stratonovich stochastic Taylor expansions as proposed in ~\cite{gao2020improved} for the case of imperfect measurements.}
Further research lines include  providing a rigorous analytic study for the stabilization property observed numerically and extending our results to the case $[H,L] \neq 0$.
%In this paper, we study a quantum projection filtering approach for an open quantum system with measurement imperfections. {\textcolor{blue}{Exact solution of the filter equation is expressed in terms of simplified stochastic differential equations by the use of a refined projection submanifold. Local approximation errors of the projection process are obtained by imposing some technical assumptions. A bound for the average total residual norm is derived. Almost sure convergence of the projected filter to one of the set of equilibrium states is proved.}} Simulations of a two-level  system are provided with the aim of verifying the efficiency of the projection filtering method in the feedback stabilization design.
%In future work, we aim at providing a rigorous analytic study for the stabilization property observed numerically. An extension of our results to the case {$[H,L] \neq 0$} is also considered in our research lines.
\section*{Appendix}
The following lemma collects some standard properties of singular values.

{\bf{Lemma A.1}.} Let $A$ and $B$ be $n \times n$ matrices. Then,
\begin{itemize}
    \item {$\sum_{i=1}^{m} s_{i}(AB) \leq \sum_{i=1}^{m} s_{i}(A) s_{i}(B), 1 \leq m \leq n$;} %\\
\item{$s_{1}(AB) \leq  s_{1}(A) s_{1}(B)$;} %\\
\item{$s_{i}(AA^{\dag})=s_{i}^{2}(A);$} %\\
\item{$\sum_{i}s_{i}(AA^{\dag})={\rm Tr}(AA^{\dag}).$}
\end{itemize}
\label{lem:A1}
\section*{Acknowledgment}
This work is supported by the Agence Nationale de la Recherche projects Q-COAST ANR- 19-CE48-0003 and IGNITION ANR-21-CE47- 0015I. The authors would like to thank Sofiane Chalal and Mario Neufcourt for the helpful discussions. 
\bibliographystyle{IEEEtran}
\bibliography{biblioibrahim}

\end{document}